\documentclass[12pt]{article}

\usepackage{fullpage}
\usepackage[T1]{fontenc}
\usepackage{ae,aecompl}
\usepackage[dvips]{graphicx}
\usepackage{cancel}

\usepackage{subfig}
\usepackage[round,authoryear]{natbib}
\usepackage{color}
\usepackage{array}
\usepackage{rotating}
\usepackage{colortbl}
\usepackage{tikz}
\usetikzlibrary{patterns}

\definecolor{webgreen}{rgb}{0,0.4,0}
\definecolor{webbrown}{rgb}{0.6,0,0}
\definecolor{purple}{rgb}{0.5,0,0.25}
\definecolor{darkblue}{rgb}{0,0,0.7}
\definecolor{darkred}{rgb}{0.7,0,0}
\definecolor{darkgreen}{rgb}{0,0.7,0}

\usepackage{hyperref}
\hypersetup{colorlinks,citecolor=darkred,filecolor=black,linkcolor=darkblue,urlcolor=webgreen}
					
\usepackage{natbib}					
\usepackage{amsmath}
\usepackage{amssymb}
\usepackage{amsthm}

\usepackage{sectsty}
\sectionfont{\centering\normalfont\scshape}
\subsectionfont{\centering\normalfont}

\newtheorem{prop}{\sc Proposition}
\newtheorem{lemma}{\sc Lemma}
\newtheorem{defn}{\sc Definition}
\newtheorem*{theorem*}{\sc Theorem}
\newtheorem{theorem}{\sc Theorem}

\usepackage{cleveref}

\title{Walrasian equilibrium: an alternate proof of existence and lattice structure}
\author{\small{Komal Malik}\thanks{Komal Malik, Indian Statistical Institute, New Delhi and Shiv Nadar Institution of Eminence;  Email: \texttt{komal.malik@snu.edu.in} }}
\date{\small{\today}}
\begin{document}
\maketitle
\begin{abstract}
We consider a model of two-sided matching market with money where buyers and sellers trade
indivisible goods with the feature that each buyer has unit demand and each seller has unit supply. The result of the
existence of Walrasian equilibrium and lattice structure of equilibrium price vectors is well
known. We provide an alternate proof for the existence and lattice structure using Tarski’s
fixed point theorem.
\end{abstract}

\section{Introduction}

We consider a model of a two-sided matching market with money where buyers and sellers trade indivisible
goods. Each buyer and seller has a unit demand and unit supply, respectively. This model is known as \textit{assignment model} and a number of markets can be captured by this model. 
Examples include labour markets matching workers with jobs with
salary as side payments, and marriage markets with dowry.
Such markets have been studied in great detail. 
One of the first papers dealing with equilibrium in these markets was by \cite{gale1960}, who showed that equilibrium prices exist for such markets when agents exhibit quasilinear preferences. The result of existence was extended in \cite{shapley1972} for the assignment market to show that the set of Walrasian equilibrium price vectors forms a lattice and the core coincides with the set of Walrasian equilibriums. \cite{Kanecko1982} shows the existence of Walrasian equilibrium without the assumption of transferable utility.
Later, \cite{quinzii1984} provided a unified model of two-sided markets, and presented the existence of the Walrasian equilibrium
and the core of the game associated with the exchange model coincides with the Walrasian equilibrium under certain assumptions. The general proof of the lattice structure of the core in two-sided matching with preference ordering on both sides was shown in \cite{demange1985}, which implies that allowing for non-quasilinear utilities in an assignment game, the result that the set of Walrasian equilibrium price vectors forms a lattice holds.

The results of the existence and lattice form of the Walrasian equilibrium price vector raise curiosity as to whether the same could be proved by a fixed point theorem on a lattice. The natural choice of fixed point theorem would be Tarski's fixed point theorem which states that the set of fixed points of a monotone function between two complete lattices forms a complete lattice. 
We provide an alternate proof for the existence and lattice form of the Walrasian equilibrium price vector using Tarski’s fixed point theorem, the ideas of which are outlined as follows.
We construct a monotonic function from a complete lattice to a complete lattice. We call this \textit{Price-adjusting} function. This function is defined in such a way that the fixed points of the function coincide with the Walrasian equilibrium price vectors, and its proof uses the characterization of the Walrasian equilibrium price vector given by \cite{DebasisTalman}. Although their paper assumes quasi-linearity in their model, their characterization result does not depend on this assumption, which allows us to use it in our model. 
We use Tarski's fixed point theorem on the function to show that the Walrasian equilibrium exists and the equilibrium price vector forms a complete lattice. 

\section{Preliminaries}
The two-sided market with unit demand, and unit supply where the sellers just value the price being paid for the object she is endowed can be modeled as a market with just buyers with unit demand and potentially having zero value. We use the same notation as \cite{DebasisTalman}.

Let $N$ be a finite set of buyers, and $M=\{a,b,\dots,m \}$ be a finite set of indivisible objects. Each buyer can be assigned at most one object. The good $0$ is a dummy good that can be assigned to more than one agent. 

An \textit{allocation} $\mu$ is an ordered sequence of objects $(\mu_1, \mu_2, \dots, \mu_n)$ such that if $\mu_i= \mu_j $ then $\mu_i=\mu_j=0.$ 
 A price vector is denoted by $\mathbf{p}:= (p_a, p_b, \dots , p_m)$, where $p_x$ is the price of the object $x$. We use the notation $p_{-a}:= (p_b,p_c \dots, p_m)$ to represent the price of all objects except the price of the object $a$. Without loss of generality, let $p_0= 0.$

  Each agent has a preference over $(M \cup \{0\}) \times \mathbb{R}$. Let $R_i$ denote the preference relation over $(\mu_i, p_{\mu_i}),$ where $p_{\mu_i}$ is the price for the object $\mu_i.$
A buyer's demand depends on the price vector $\mathbf{p}.$ A buyer $i'$s demand is given by $D_i(\mathbf{p})$ where

$$D_i(\mathbf{p})= \{ x: (x,p_x) ~R_i~(y, p_y), \forall y \in M \cup \{0\} \}.$$

\begin{defn}
A Walrasian Equilibrium (WE) is a price vector $\mathbf{p}$ and a feasible allocation $\mu$ such that
$$\mu_i \in D_i(\mathbf{p}) ~  \text{for all } i \in N  \hspace{4.7cm} (\textbf{WE-1})$$ and 
$$p_j=0~ \text {for all } j \in M \text{ that are unassigned in } \mu~~~~~(\textbf{WE-2})$$

If $(\mu,\mathbf{p})$ is a WE, then $\mu$ is a Walrasian equilibrium allocation and $\mathbf{p}$ is a Walrasian equilibrium price vector.

\end{defn}

We define demanders of  set of goods $S \subseteq M$ at price vector $\mathbf{p}$ as
$$U(S, \mathbf{p})= \{i \in N: D_i(\mathbf{p}) \cap S \ne \emptyset \}.$$

We define exclusive demanders of set $S \subseteq M$ at price vector $\mathbf{p}$ as 
$$ O(S, \mathbf{p})= \{ i \in N: D_i(\mathbf{p}) \subseteq S \}.$$

Let $M^+(\mathbf{p})$ be the set of goods that have strictly positive prices.
\begin{defn}
A set of goods S is (weakly) underdemanded at price vector $\mathbf{p}$ if $S \subseteq M^+(\mathbf{p}) $ and $\left| U(S,p) \right|< (\le) \left| S\right|$
\end{defn}

\begin{defn}
A set of goods S is (weakly) overdemanded at price vector $\mathbf{p}$ if $S \subseteq M $ and $\left| O(S,p) \right|> (\ge) \left| S\right|$.
\end{defn}

\begin{defn}
A good $x$ is minimally overdemanded (underdemanded) at price vector $\mathbf{p}$ if there exists a set $S$ containing $x$ which is overdemanded (underdemanded) but for all sets $T \subset S$, containing $x$, $T$ is not overdemanded (underdemanded).
\end{defn}

Let $\mathbf{H}= (H,H, \dots)$ with $H= \max_{x \in M} \max_{i \in N} WP(\{x\},0; R_i)$. $H$ is high enough price for which no good is minimally overdemanded.\footnote{ Since the price of each good is $H$, each agent would demand good 0. As a result, exclusive demanders of any set $S$ s.t. $0 \notin S$ would be zero and thus, no set would be overdemanded.}

\subsection{Overdemand and Underdemand}

We define two \textit{critical/tipping} prices as follows:
\begin{enumerate}
\item $\sup ~\mathcal{O}_a(p_{-a})=\sup \big[ \{ p_a: a \textrm{ is minimally overdemanded at }(p_a,p_{-a})\} \cup \{0\}\big]$
\item $\inf ~\mathcal{U}_a(p_{-a})=\inf \{ p_a: a \textrm{ is minimally underdemanded at }(p_a,p_{-a})\} $ 
\end{enumerate}

\paragraph{}
Note that $\sup ~\mathcal{O}_a(p_{-a})$ is highest price for which $a$ could be minimally overdemanded. For all $p_a>\sup ~\mathcal{O}_a(p_{-a}), $ $a$ is not minimally overdemanded. Similarly, $\inf ~\mathcal{O}_a(p_{-a}) $ is the lowest price at which $a$ could be underdemanded. For all $p_a< \inf ~\mathcal{U}_a(p_{-a})$, $a$ would not be minimally overdemanded. \\\\
{\sc Fact :} For all $p_{-a} \in \mathbb{R}^{N-1}$, $ \inf ~\mathcal{U}_a(p_{-a}) \ge \sup ~\mathcal{O}_a(p_{-a}).$ \\
Suppose not, for contradiction i.e. $\inf ~\mathcal{U}_a(p_{-a}) < \sup ~\mathcal{O}_a(p_{-a})$ for some $p_{-a}$. Let $p \in( \inf ~ \mathcal{U}_a(p_{-a}) ,\sup ~\mathcal{O}_a(p_{-a})).$ But this implies that  $a$ is minimally underdemanded and minimally overdemanded  at $(p, p_{-a})$. A contradiction. 
\paragraph{}

The critical points defined above play a crucial role in defining a {\textit{price adjusting function}}.
We use tipping prices defined above to define the two \textit{neutral} prices - $\mathcal{S}_a(p_{-a})$ and $ \mathcal{I}_a(p_{-a})$. These prices are termed \textit{neutral} because the constructed \textit{price adjusting function}  is such that for all prices, $p_a \in [ \mathcal{S}_a(p_{-a}), \mathcal{I}_a(p_{-a})]$, $f_a(p_a, p_{-a})=p_a$. 
\paragraph{}
We define
\begin{enumerate}
\item 
$\mathcal{S}_a(p_{-a})= \inf ~\{ q_a: \textrm{ no object is overdemanded at  }(q_a, q_{-a}) \textrm{ where }  q_a \ge \sup ~\mathcal{O}_a(q_{-a}), ~q_{-a} \ge p_{-a}\} \text{ and }$
\item 
$\mathcal{I}_a(p_{-a})= \inf[\{q_a: \textrm{ no object is underdemanded at } (q_a, q_{-a}) \textrm{ where }q_a \ge \inf ~\mathcal{U}_a(q_{-a}), ~q_{-a} \ge p_{-a} \} \cup \{H\}]$
\end{enumerate}

The neutral prices $\mathcal{S}_a(p_{-a})$ and $\mathcal{I}_a(p_{-a})$ are well defined. We argue this below.

The set $\{ q_a: \textrm{ no object is overdemanded at  }(q_a, q_{-a}) \textrm{ where }  q_a \ge \sup ~\mathcal{O}_a(q_{-a}), ~q_{-a} \ge p_{-a}\}$ is non-empty. For high enough price $\mathbf{H}$, only good 0 would be demanded and no object would be overdemanded. Also, it is lower bounded by construction. Thus, $\mathcal{S}_a(p_{-a})$ is well defined.
Though, the set $\{q_a: \textrm{ no object is underdemanded at } (q_a, q_{-a}) \textrm{ where }q_a \ge \inf ~\mathcal{U}_a(q_{-a}), ~q_{-a} \ge p_{-a} \}$ may be empty. Adding $H$ to this set ensures that infimum is applied over a non-empty set. It is also lower bounded, by construction. Thus, $\mathcal{I}_a(p_{-a})$ is also well defined.
\paragraph{}
Finally, we define \textit{price-adjusting} function $\mathbf{f}: [\mathbf{0},\mathbf{H}] \rightarrow [\mathbf{0},\mathbf{H}]$ as follows:	
\begin{equation*}
	f_a(\mathbf{p})= \left \{
\begin{array}{l l}
	\mathcal{I}(p_{-a}) & \textrm{ if }p_a > \mathcal{I}_a(p_{-a}) \ge \mathcal{S}_a(p_{-a})  ~~~-~(I)\\
	p_a & \textrm{ if }\mathcal{S}_a(p_{-a}) \le p_a \le \mathcal{I}_a(p_{-a}) ~~~-~( II)\\
	\mathcal{S}_a(p_{-a}) & \textrm{ if } p_{a} < \mathcal{S}_a(p_{-a}) \le \mathcal{I}_a(p_{-a}) ~~~-(III)\\
	\frac{ \mathcal{S}_a(p_{-a}) + \mathcal{I}_a(p_{-a})}{2} & \text{ if }  \mathcal{I}_a(p_{-a}) < \mathcal{S}_a(p_{-a}) ~~~~~~~~~~-(IV)
	\end{array}  \right.
	\end{equation*}

To visualise how $p_a$ maps to $f_a(\mathbf{p})$, given price vector $p_{-a}$, we look at bar plot as shown in the figure below.%~\ref{fig:eg1}.
\footnote{Note that it is plotted for a given $p_{-a}$.} Given $p_{-a},$ we would have $\mathcal{I}_a(p_{-a}), \mathcal{S}_a(p_{-a}) \in [0,H].$ The figure $1(a)$ displays the case where $H>\mathcal{I}_a(p_{-a}) > \mathcal{S}_a(p_{-a})> 0.$

\begin{figure}[h]
    \centering

\subfloat[]{{\includegraphics[width=2cm]{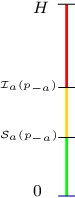} }}%
    \qquad
    \subfloat[]{{\includegraphics[width=2cm]{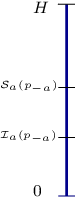}} }%
    
    \caption{Mapping of Price Adjusting function}
 \label{fig:eg1}
 
\end{figure}

For all $p_a$ in red region, the value of the function is lower than $p_a$ and equals $\mathcal{I}_a(p_{-a}).$
For all $p_a$ in green region, the value of the function is higher than $p_a$ and equals $\mathcal{S}_a(p_{-a}).$
For all $p_a$ in yellow region, the value of the function is equal to $p_a$.

Figure $1(b)$ shows the case where $\mathcal{S}_a(p_{-a})> \mathcal{I}_a(p_{-a})$. All the prices $p_a$ would lie in blue region and $f_a(p_a,p_{-a})= \frac{\mathcal{S}_a(p_{-a}) + \mathcal{I}_a(p_{-a})}{2}$.

From now on, all the points satisfying conditions $(I), ~(II), ~(III)$ and $(IV)$ are represented by red, yellow, green and blue colour, respectively.

\section{Result}

The Tarski's fixed-point theorem states that if F is a monotone
function on a complete lattice, the set of fixed points of F
forms a complete lattice. \\
The \textit{price adjusting} function defined above is monotonic (Proposition \ref{prop:1}) and the set of fixed points is equal to the set of Walrasian equilibrium price vectors (Proposition \ref{prop 2}). By Tarski's fixed-point theorem, it follows that the set of Walrasian equilibrium price vectors forms a complete lattice.

We provide the proofs of Propositions \ref{prop:1}
and \ref{prop 2} below. 
	
\begin{prop} \label{prop:1}
$\mathbf{f}(\mathbf{p})$ is a monotonic function.
\end{prop}

\begin{proof}
Suppose $\mathbf{p} \le \mathbf{q}$. We show that for all $ a \in M, $ $ f_a(\mathbf{p}) \le f_a(\mathbf{q})$. 

We start by proving useful lemmas.
\begin{lemma} 
\label{lem:1}
For  $\mathbf{p} \le \mathbf{q},$  we have $\mathcal{S}_a(p_{-a}) \le \mathcal{S}_a(q_{-a})$ and $\mathcal{I}_a(p_{-a}) \le \mathcal{I}_a(q_{-a})$. 
\end{lemma}
\begin{proof}
Suppose $ \mathbf{p} \le \mathbf{q}.$ Let $h_a  \in \{ h_a: \textrm{ no object is overdemanded at  }( h_a, h_{-a}) \textrm{ where }  h_a \ge \sup ~\mathcal{O}_a(p_{-a}), h_{-a} \ge q_{-a} \}$. Note that $h_a \ge \sup ~\mathcal{O}_a(q_{-a}) \ge \sup  \mathcal{O}_a(p_{-a})$, where last inequality follows from monotonicity of $\sup ~\mathcal{O}_a.$ Also, $h_{-a} \ge q_{-a} \ge p_{-a}$, where the last inequality follows from the fact that $\mathbf{p} \le \mathbf{q}.$ This implies $h_a \in \{ h_a: \textrm{ no object is  overdemanded at  }( h_a, h_{-a}) \textrm{ where } h_a \ge \sup ~\mathcal{O}_a(h_{-a}),  h_{-a} \ge p_{-a} \}.$ 

Thus, we have 
\begin{align*}
 &\{ h_a: \textrm{ no object is overdemanded at  }( h_a, h_{-a}) \textrm{ where }  h_a \ge \sup ~\mathcal{O}_a(p_{-a}), h_{-a} \ge q_{-a} \}\\
 & \subseteq \{ h_a: \textrm{ no object is overdemanded at  }( h_a, h_{-a}) \textrm{ where } h_a \ge \sup ~\mathcal{O}_a(h_{-a}),  h_{-a} \ge p_{-a} \}\\
 \Rightarrow & \inf ~\{ h_a: \textrm{ no object is overdemanded at  }( h_a, h_{-a}) \textrm{ where }  h_a \ge \sup ~\mathcal{O}_a(p_{-a}),h_{-a} \ge q_{-a} \} \equiv \mathcal{S}_a(q_{-a})\\
 & \ge \inf ~\{ h_a: \textrm{ no object is overdemanded at  }( h_a, h_{-a}) \textrm{ where }  h_a \ge \sup ~\mathcal{O}_a(p_{-a}),h_{-a} \ge p_{-a} \} \equiv \mathcal{S}_a(p_{-a}).
\end{align*}
Analogously, we can show that $\mathcal{I}_a(p_{-a}) \le \mathcal{I}_a(q_{-a}).$
\end{proof}

\begin{lemma}
\label{lem:2} For all $\ell_ \le h_a$, for all $\ell_{-a}$, we have $f_a(\ell_, \ell_{-a}) \le f_a(h_a, \ell_{-a}).$ 
 \end{lemma}
\begin{figure}[h]
    \centering
    \subfloat{{\includegraphics[width=2cm]{ImagesWE/RYG1.eps} }}%
    \qquad
    \subfloat{{\includegraphics[width=2cm]{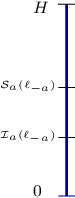} }}%
    \qquad

    \subfloat{{\includegraphics[width=7cm]{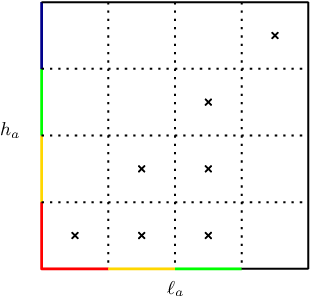} }}%

    \caption{Coloured possibility matrix, given $\ell_{-a}$}
    \label{fig:prop1}%
\end{figure}

\begin{proof}
Fix $\ell_{-a}.$ This would fix $\mathcal{I}_{a}$ and $\mathcal{S}_{a}$
at $\ell_{-a}.$ 

We start by introducing a coloured matrix.
The matrix reflects the possibility of $h_a$ and $\ell_a$ lying in a particular region. The $x-axis$ is $\ell_a$ and $y-axis$ is $h_a$. The cell in y coloured row and x coloured column would reflect if there is a possibility that $h_a$ lies in x coloured region and $\ell_a$ lies in y coloured region. Since $\ell_a \le h_a$, the only possible cases are those marked with a cross.

Now we argue that $f_a(\ell_a, \ell_{-a}) \le f_a(h_a, \ell_{-a})$.

If $\ell_a$ lies in green region, $f_a(\ell_a, \ell_{-a})= \mathcal{S}_a(\ell_{-a})$. ~$h_a$ can lie in any of the region and $f_a(h_a, \ell_{-a}) \ge \mathcal{S}_a(\ell_{-a}).$ 

If $\ell_a$ lies in yellow region, $f_a(\ell_a, \ell_{-a})= \ell_a.$ ~$h_a$ can lie in yellow region or red region. This gives $f_a(h_a, \ell_{-a}) \ge \min\{\mathcal{I}_a(\ell_{-a}), h_a\} \ge \ell_a,$ where the former inequality follows from definition of $\mathbf{p}$ and the latter inequality follows from the fact that $h_a \ge \ell_a$ and $\ell_a$ lies in yellow region.

If $\ell_a$ lies in red region, $f_a(\ell_a, h_{-a})= \mathcal{I}_a( \ell_{-a})$. ~$h_a$ can lie in red region only. This gives $f_a(h_a, \ell_{-a})=\mathcal{I}_a( \ell_{-a})$.

If $\ell_a$ lies in blue region, $f_a(\ell_a,\ell_{-a})= \frac{\mathcal{S}_a(\ell_{-a}) + \mathcal{I}_a(\ell_{-a})}{2}$. Also, $h_a$ lie in blue region, implying that $f_a(h_a, \ell_{-a})=  \frac{\mathcal{S}_a(\ell_{-a}) + \mathcal{I}_a(\ell_{-a})}{2}.$

Thus, in all cases, $f_a(\ell_a, \ell_{-a}) \le f_a(h_a, \ell_{-a}).$
\end{proof}
\medspace
Lemma \ref{lem:2} has shown monotonicity in the first argument. Now we show monotonicity in other arguments, i.e. $f_a(p_a, p_{-a}) \le f_a(p_a, q_{-a})$.
Since $p_{-a} \le q_{-a},$ by Lemma \ref{lem:1}, we get the possibility matrix as the following:

\begin{figure}
\centering
\includegraphics[width=9cm]{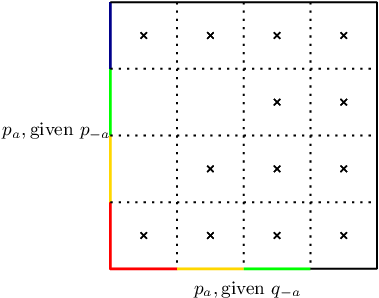}
\caption{Coloured possibility matrix, given $p_{-a}$ and $q_{-a}$}
\hspace{1in}
\end{figure}

{\sc Case 1. } $p_a$ lies in blue region for $p_{-a}$.

We have $f_a(p_a,p_{-a})= \frac{\mathcal{S}_a(p_{-a}) + \mathcal{I}_a(p_{-a})}{2}.$
If $p_a$ is blue region for $q_{-a}$ as well, then $f_a(p_a, q_{-a})= \frac{\mathcal{S}_a(q_{-a}) + \mathcal{I}_a(q_{-a})}{2} \ge  \frac{\mathcal{S}_a(p_{-a}) + \mathcal{I}_a(p_{-a})}{2}=f_a(p_a,p_{-a}), $ where inequality follows from Lemma \ref{lem:1}.
If $p_a$ not in blue region for $q_{-a}$, then $f_{a}(p_a, q_{-a}) \ge \mathcal{S}_a(q_{-a}) \ge \mathcal{S}_a(p_{-a}) \ge \frac{\mathcal{S}_a(p_{-a}) + \mathcal{I}_a(p_{-a})}{2}= f_a(p_a,p_{-a}),$ where the first inequality follows from definition of $f_a$, ~second inequality follows from Lemma \ref{lem:1} and third inequality follows from the fact that $\mathcal{S}_a(p_{-a}) \ge \mathcal{I}_a(p_{-a}).$

{\sc Case 2. } $p_a$ lies in blue region for $q_{-a}.$

We have $f_a(p_a, q_{-a})= \frac{\mathcal{S}_a(q_{-a}) + \mathcal{I}_a(q_{-a})}{2}.$
$f_a(p_a, p_{-a}) \le \mathcal{I}_a(p_{-a}) \le \mathcal{I}_a (q_{-a}) \le \frac{\mathcal{S}_a(q_{-a}) + \mathcal{I}_a(q_{-a})}{2}= f_a(p_a, q_{-a})$, where the second inequality follows from Lemma $\ref{lem:1}$ and the third inequality follows from the fact that $\mathcal{S}_a(q_{-a})> \mathcal{I}_a(q_{-a}).$

Notice that in the reduced matrix without the blue region, one just needs to show the inequality for the diagonal. In the lower triangular matrix, it is trivial from the mapping that the inequality will be satisfied.

On the diagonal, the value of the function,
$(f_a(p_a, p_{-a}), f_a(p_a, q_{-a}))$ is $(\mathcal{I}_a(p_{-a}), \mathcal{I}_a(q_{-a})), ~ (p_a, p_a)$ and $(\mathcal{S}_a(p_{-a}), \mathcal{S}_a(q_{-a}))$ for red, yellow and green region, respectively. By Lemma \ref{lem:1}, for all these possibilities, we have $f_a(p_a, p_{-a}) \le f_a(p_a, q_{-a})$. 

Thus, we have established that $f_a(p_a, p_{-a}) \le f_a(p_a, q_{-a}).$ This combined with Lemma \ref{lem:2} gives $f_a(p_a, p_{-a}) \le f_a(q_a, q_{-a}).$ 
\end{proof}

\begin{prop} \label{prop 2}
The following statements are equivalent:
\begin{enumerate}
\item $\mathbf{p}$ is fixed-point of $\mathbf{f}.$
\item $\mathbf{p}$ is Walrasian price vector.
\end{enumerate}
\end{prop}

\begin{proof}

{\sc Necessity: }
Suppose $\mathbf{p}$ is a Walrasian equilibrium. We would show that $\mathbf{f}(\mathbf{p})= \mathbf{p}$.\\
Note that since $\mathbf{p}$ is Walrasian equilibrium, no object $a \in M$ will be minimally underdemanded or overdemanded.\\
Fix an arbitrary object $a \in M$. We would argue that $\mathcal{I}_a (p_{-a}) \ge p_a \ge \mathcal{S}_a(p_{-a})$.\\
We have $\mathcal{I}_a (p_{-a}) \ge \inf  \mathcal{U}_a(p_{-a}) \ge p_a$, where the former inequality follows by definition of $\mathcal{I}_a$ and the latter inequality follows from the fact that no object is minimally underdemanded.

We show that $\mathcal{S}_a(p_{-a}) \le p_a$. Suppose not for contradiction. Also note that $p_a \ge \sup ~\mathcal{O}_a(p_{-a})$ follows from the fact that no good is overdemanded. Combining the last two inequalities, we get $\mathcal{S}_a(p_{-a})> p_a \ge \sup ~\mathcal{O}_a(p_{-a})$. By definition of $\mathcal{S}_a$, this implies some object was minimally overdemanded at $(p_a, p_{-a}).$ A contradiction.

Thus, we have  $\mathcal{I}_a (p_{-a}) \ge p_a \ge \mathcal{S}_a(p_{-a})$. By definition of $\mathbf{f}$, $f_a(\mathbf{p})= p_a.$
Since, for all $a \in M $, $f_a(\mathbf{p})= p_a$, we have $\mathbf{f}(\mathbf{p})= \mathbf{p}.$
\vspace{0.02in}\\
\noindent {\sc Sufficiency: }
Suppose $\mathbf{p}$ is a fixed-point. 
Note that by definition, $\mathcal{I}_a(p_{-a} ) \ge  \inf ~\mathcal{U}_a(p_{-a})$ and $\mathcal{S}_a(p_{-a}) \ge \sup ~\mathcal{O}_a(p_{-a})$.

We will use the characterization of WE price vector in Theorem 1  of \cite{DebasisTalman}, to show that $\mathbf{p}$ is Walrasian equilibrium price vector. The theorem is stated below.

\begin{theorem*}
    A price vector $\mathbf{p}$ is a WE price vector if and only if no set of goods is overdemanded and no set of goods is underdemanded at $\mathbf{p}.$
\end{theorem*}

We first show that no set is overdemanded at price $\mathbf{p}$.
Suppose for contradiction that at $\mathbf{p}$, there exists an overdemanded set $S$. Let $a \in S$ is minimally overdemanded at $\mathbf{p}$. 

Now we argue that $\mathcal{S}_a(p_{-a})> p_a.$ If for all $x \in S$, $h_x \in [p_x, p_x +\epsilon]$ where $\epsilon \rightarrow 0$  and for  all $x \notin S$, $h_x= p_x$, then the set $S$ will be overdemanded at $\mathbf{h}$.\footnote{Very small increase in prices would not change demand for exclusive demanders and thus, the set would still be overdemanded.} As a result, $\mathcal{S}_a(p_{-a}) > p_a$. 

Thus, only condition $(III)$ or $(IV)$ could hold true. By condition $(III) $, $f_a(\mathbf{p})= \mathcal{S}_a(p_{-a})> p_a$. A contradiction to the assumption that $\mathbf{p}$ is fixed-point.
Thus, the only possibility that we are left with is that $f_a(\mathbf{p})= \frac{\mathcal{S}_a(p_{-a}) + \mathcal{I}_a(p_{-a})}{2}$. But %since $a$ is overdemanded, 
$f_a(\mathbf{p})=p_a= \frac{\mathcal{S}_a(p_{-a}) + \mathcal{I}_a(p_{-a})}{2} >\sup ~\mathcal{O}(p_{-a})$. A contradicton to our assumption that $a$ is minimally overdemanded. 

Thus, no set is overdemanded at $\mathbf{p}$.

Now we show that no set is overdemanded at $\mathbf{p}.$ Suppose for contradiction that set $T$ is overdemanded at $\mathbf{p}.$ Let $b \in T$ be the minimally underdemanded good. 
For all $x \in T$, if $h_x \in [p_x- \epsilon, p_x]$ \ where $\epsilon \rightarrow 0$ , and for all $x \notin T$, $h_x= p_x$, the set $T$ will be underdemanded\footnote{Note that by the definition of underdemanded set, $p_x> 0$}. Thus, $\mathcal{I}(p_{-b})< p_b$. 

Thus, condition $(II)$ or $(III)$ cannot be satisfied. If condition $(I)$ is satisfied, then $f_b(\mathbf{p}) = \mathcal{I}(p_{-b})< p_b.$ A contradiction.

Thus, we are left with the case where $f_b(\mathbf{p})= \frac{\mathcal{S}_b(p_{-b}) + \mathcal{I}_b(p_{-b})}{2}= p_b$. This implies that $\mathcal{S}_b(p_{-b})>p_b \ge \inf ~\mathcal{U}_b(p_{-b}) \ge \sup ~\mathcal{O}_b(p_{-b})$, where first inequality follows from the fact that $p_b$ is average of $\mathcal{I}_b(p_{-b})$ and $\mathcal{S}_b(p_{-b})$ and $\mathcal{S}_b(p_{-b}) >\mathcal{I}_b(p_{-b}) $ . By definition of $\mathcal{S}_b(p_{-b})$, some object is minimally overdemanded at $(p_b, p_{-b})$. 

However we have already shown that no object is minimally overdemanded at fixed-point, $\mathbf{p}$. 

Thus, no set is overdemanded or underdemanded at $\mathbf{p}.$ By Theorem 1  of \cite{DebasisTalman}, $\mathbf{p}$ is Walrasian equilibrium price vector.
\end{proof}

Combining Propositions \ref{prop:1} and \ref{prop 2}, and Tarski's fixed-point theorem, we establish the well-known result:

\begin{theorem}
    The set of Walrasian equilibrium price vectors forms a complete lattice.
\end{theorem}

\section*{Acknowledgement}
We thank Debasis Mishra for the useful conversations and suggestions. 

\section*{Funding}
This research did not receive any specific grant from funding agencies in the public, commercial, or not-for-profit sectors.

\bibliographystyle{ecta}
\bibliography{WalrasOrder}
\newpage

\end{document}